\documentclass[11pt]{article}
\usepackage{amsmath}
\usepackage{amsthm}
\usepackage{amsfonts}
\usepackage{cite}
\usepackage{hyperref}
\input{Qcircuit}

\newcommand{\qs}[1]
{|#1\rangle}

\newcommand{\qsl}[1]
{\langle #1 |}

\newcommand{\qinner}[2]
{\langle #1 | #2 \rangle }

\begin{document}
\theoremstyle{plain}
\newtheorem{theorem}{Theorem}
\newtheorem{lemma}{Lemma}
\newtheorem{claim}{Claim}
\newtheorem{corollary}{Corollary}
\newtheorem{remark}{Remark}


\newcommand{\finpr}{\hfill $\square$ \vspace{2mm}}

\def\be{\begin{eqnarray}}

\def\ee{\end{eqnarray}}

\def\bee{\begin{eqnarray*}}

\def\eee{\end{eqnarray*}}

\title{ \vspace{-2cm} Commuting quantum circuits: efficient classical simulations versus hardness results}
\author{Xiaotong Ni\footnote{xiaotong.ni@mpq.mpg.de} \ \  and \ Maarten Van den  Nest\footnote{maarten.vandennest@mpq.mpg.de} \\ \\ {\normalsize Max-Planck-Institut f\"ur Quantenoptik,} \\ {\normalsize Hans-Kopfermann-Stra\ss e 1, D-85748 Garching, Germany. }}

\maketitle

\begin{abstract}

The study of quantum circuits composed of commuting gates is particularly useful to understand the delicate boundary between quantum and classical computation. Indeed, while being a restricted class, commuting circuits exhibit genuine quantum effects such as entanglement. In this paper we show that the computational power of commuting circuits exhibits a surprisingly rich structure. First we show that every 2-local commuting circuit acting on $d$-level systems and followed by single-qudit measurements can be efficiently simulated classically with high accuracy. In contrast, we  prove that such strong simulations are hard for 3-local circuits. Using sampling methods we further show that all commuting circuits composed of exponentiated Pauli operators $e^{i\theta P}$ can be simulated efficiently classically when followed by single-qubit measurements. Finally, we show that commuting circuits can efficiently simulate certain non-commutative processes, related in particular to constant-depth quantum circuits. This gives evidence that the power of commuting circuits goes beyond classical computation.
\end{abstract}

\section{Introduction}

Since the discovery of Shor's factoring algorithm \cite{shor1999polynomial}, the question whether quantum computers possess exponentially more power then classical computers has been one of the central problems in the field. Similar to other notorious problems in computational complexity theory, this question is very difficult. For example, a proof that $\text{P}\neq \text{BQP}$ would imply that $\text{P}\neq \text{PSPACE}$, which is a longstanding open problem. A useful approach to gain insight into the relationship between quantum and classical computing power is to study restricted classes of quantum circuits and analyze their  power.  For several restricted but nontrivial classes of quantum circuits, it has been found that efficient classical simulations are possible. For instance, if in each step of a quantum circuit the entanglement (quantified by the $p$-blockedness \cite{Jozsa2002Role} or by the Schmidt rank \cite{Vidal2003Efficient}) is bounded, the circuit can be simulated efficiently classically. Such results demonstrate that certain types of entanglement must be generated in sufficiently large amounts if a quantum algorithm is to yield an exponential speed-up. Certain other circuit classes can be simulated classically using entirely different arguments not based on entanglement considerations  \cite{Gottesman1997Stabilizer, Valiant2002Quantum, Knill2001Fermionic, Jozsa2008Matchgates, VandenNest2010Simulating, VandenNest2012Efficient}, e.g. by using the Pauli stabilizer formalism \cite{Gottesman1997Stabilizer} or the framework of matchgate tensors \cite{Valiant2002Quantum, Knill2001Fermionic, Jozsa2008Matchgates} .

Conversely, it has been shown that some restricted quantum computation schemes can perform tasks that appear to be hard classically \cite{Aaronson2011computational, Bremner2010Classical, Jordan2010permutational}. For example, in \cite{Aaronson2011computational} the hardness of simulating linear optical quantum computation was discussed. In \cite{Bremner2010Classical} it was shown that simulating the output probability distribution of commuting quantum circuits would imply a collapse of the polynomial hierarchy and is thus highly unlikely. Besides theoretical importance, these results also lower the threshold to demonstrate nontrivial quantum computation in experiments.

In this paper we  focus on commuting quantum circuits. Several features make such circuits interesting. For example, commuting circuits exhibit genuine quantum effects, e.g. they can generate highly entangled states (such as cluster states \cite{Briegel2001persistent}). Further, since commuting operations can be performed simultaneously, there is no time order in the computation, which is drastically different from other computational models. Moreover, all gates in the circuit can be diagonalized simultaneously. The latter property might at first sight suggest an intrinsic simplicity of this circuit class; however it is important to note that the diagonalizing unitary can by a complex entangling operator. In \cite{Bremner2010Classical} as well as in the present paper evidence is given that commuting circuits indeed have nontrivial power beyond classical computation. It is also interesting to note that commuting operations have recently caught attention in different areas as well, such as the study of local Hamiltonian problem \cite{Bravyi2004Commutative,Aharonov2011Complexity, Schuch2011complexity, Hastings2012trivial}.

Compared to earlier work \cite{Shepherd2009instantaneous, Bremner2010Classical,Shepherd2010binary} which considered commuting gates that can be diagonalized in a local basis, we will consider general commuting gates acting on $d$-level systems. We will show that the computational power of commuting circuits exhibits a surprisingly rich structure. For example, the degree of hardness varies significantly depending on whether the gates are  2-local or 3-local. This indicates that commuting quantum circuits might serve as a interesting intermediate class between classical and universal quantum circuits. 

Our main results can be summarized as follows (here the terms "strongly" and "weakly" specify different notions of classical simulation, to be defined below):
\begin{itemize}
\item {\it 2-local circuits are easy.} All uniform families of commuting circuits consisting of 2-local gates acting on $d$-level systems and followed by a single measurement can be strongly simulated by classical computation, for every $d$.
\item {\it 3-local circuits are hard. }Uniform families of commuting circuits consisting of 3-local gates acting on qubit systems and followed by a single measurement cannot be strongly simulated by classical computation, unless every problem in $\#$P has a polynomial-time classical algorithm.
\item {\it Commuting Pauli circuits.} All uniform families of commuting circuits consisting of exponentiated Pauli operators $e^{i\theta P}$ and followed by a single-qubit measurement can be efficiently simulated classically weakly. Furthermore, even when such circuits display a small degree of non-commutativity, an efficient classical simulation remains possible.
\item {\it Mapping non-commuting circuits to commuting circuits. } Certain non-commuting quantum processes (related to bounded-depth circuits) can be efficiently simulated by purely commuting quantum circuits.
\end{itemize}
Finally, it is noteworthy that several distinct techniques were used to prove the above  results, including tensor network methods, sampling methods as well as the Pauli stabilizer formalism. This is also an illustration of the rich structure displayed by commuting quantum circuits.
\section{Preliminaries}

\subsection{Commuting quantum circuits}
The quantum circuits considered in this work will always be unitary. The size of a circuit is the number of gates of which it consists. A $d$-level Hilbert space will sometimes generically be called a ``qudit''. For an operator $A$ acts on a system of $n$ qudits labeled by $1\cdots n$, the support of $A$ is the subset of qudits on which it acts nontrivially. A quantum circuit acting on $n$ qudits is said to be $k$-local if the support of each of its gates contains at most $k$ qudits. A family of $n$-qubit quantum circuits ${\cal C}_n$ has \emph{polynomial size} $m$ if $m$ scales polynomially with $n$, denoted by $m=$ poly$(n)$.

A \emph{commuting circuit} is a quantum circuit consisting of pairwise commuting gates. A $k$-local commuting circuit is in \emph{standard form} if for every subset $S\subseteq\{1, \dots, n\}$ consisting of $k$ qudits there is at most one gate $G_i$ with supp$(G_i)\subseteq S$. A $k$-local commuting circuit ${\cal C}= G_m\cdots G_1$ in standard form contains at most $n \choose k$ gates, so that the size of the circuit scales polynomially with $n$ if $k$ is constant. For example, a two-local commuting circuit is in standard form if for every $i, j=1\cdots n$ with $i<j$ there is at most one gate in the circuit with support contained in $\{i, j\}$; such circuit has size $O(n^2)$. Every $k$-local commuting circuit can be brought into standard form by replacing all gates in the circuit with support contained in $S$ by a single gate given by the total product of these gates, for every subset $S$ consisting of $k$ qudits. Furthermore if $k$ is constant and if the original circuit has size $m$ then this procedure to bring a circuit in normal form can be carried out efficiently i.e. in poly$(n, m)$ steps.

A simple example of a commuting circuit is ${\cal C} = G_m\cdots G_1$ with gates  \be \label{example_comm}G_i = U\otimes U D_i U^{\dagger}\otimes U^{\dagger},\ee where $U$ is a fixed single-qudit unitary operator (independent of $i$) and where each $D_i$ is a diagonal unitary operator. In other words each gate is diagonal in the same local basis. This  class of commuting circuits has been considered in \cite{Shepherd2009instantaneous,Bremner2010Classical}.

By their commutativity, all gates in any commuting circuit can be diagonalized simultaneously i.e. there exists a unitary operator $V$ such that $VG_iV^{\dagger}$ is diagonal for every gate $G_i$. In the example (\ref{example_comm}), the diagonalizing operator is a simple tensor product $V=U\otimes\cdots\otimes U$. This example does however not represent the most general situation since $V$ may be a global, entangling operation---even when the commuting circuit is $k$-local with $k$ constant. Consider for example an $n$-qubit 3-local circuit with gates  $G_j= e^{i\theta_j K_j}$ where the commuting  operators $K_j$ are defined as follows:\be\label{K_k} K_1 = X_1Z_2,\quad  K_i = Z_{i-1} X_i Z_i, \quad K_n= Z_{n-1}X_n, \quad \mbox{with}\quad i=2, \dots n-1.\ee Here $Z_i$ and $X_i$ denote the Pauli $X$ and $Z$ operators acting on qubit $i$. The operators $K_j$ are the well known stabilizers of the 1D cluster state. Let $H$ denote the Hadamard gate and let CZ $=$ diag$(1, 1, 1, -1)$ denote the controlled-$Z$ gate. It is then easily verified that the entangling operation \be V = H^{\otimes n}\prod_{i=1}^{n-1} CZ_{i, i+1}\ee sends $K_j\to VK_jV^{\dagger} = Z_j$ and thus simultaneously diagonalizes the $K_j$. Furthermore it can be shown that no tensor product of single-qubit operations can perform such a diagonalization .

The example (\ref{K_k}) shows that there exist $k$-local commuting circuits where the diagonalizing unitary $V$ is a global, entangling operator. Nevertheless this example is still rather well-behaved as $V$ can be computed efficiently and moreover has a relatively simple structure. In fact in section \ref{sect_pauli} we will investigate commuting circuits composed of exponentiated Pauli operators $e^{i\theta P}$ in more detail and show that such circuits have efficient classical simulations (relative to certain measurements). For general $k$-local commuting circuits, however, the unitary $V$ may have a complex structure and be computationally difficult to determine. This feature is in part responsible for the complexity of commuting quantum circuits.

\subsection{Classical simulations of quantum circuits}\label{sect_class_simulation}

There are several valid notions of efficient classical simulations of quantum circuits. Two notions will be considered in this work viz. \emph{strong} and \emph{weak} simulations. Their main difference lies in the \emph{accuracy} achieved in the classical simulation: roughly speaking, strong simulations achieve an exponential precision whereas weak simulations achieve polynomial precision. We mainly follow the definitions of ~\cite{Bremner2010Classical}.

Consider a uniform family of  $k$-local $n$-qubit quantum circuits ${\cal C}_n$ for some constant $k$.  The input states are standard basis states. The circuits are followed by measurement of the Pauli observable  $Z$ on the first qubit. The expectation value is denoted by $\langle Z_1\rangle$.

\

\noindent{\bf Strong simulations.} We say that ${\cal C}_n$ can be efficiently simulated classically in the strong sense if there exists a classical algorithm with runtime poly$(n, \log \frac{1}{\epsilon})$  which outputs a number $E$ such that \be\label{def_strong} | E - \langle Z_1\rangle| \leq \epsilon.\ee  Thus a strong simulation algorithm achieves an exponential accuracy $\epsilon = 2^{-\mbox{\scriptsize{poly}}(n)}$  in poly$(n)$ time.

\

\noindent{\bf Weak simulations.} We say that ${\cal C}_n$ can be efficiently simulated classically in the weak sense if there exists a  classical algorithm with runtime poly$(n, \frac{1}{\epsilon})$ which outputs a number $E$ satisfying (\ref{def_strong}). Thus a weak simulation algorithm achieves  \emph{polynomial} accuracy $\epsilon = 1/$poly$(n)$ in polynomial time. We will often allow weak simulations to fail with an exponentially small probability. In this sense, we say that ${\cal C}_n$ can be efficiently simulated classically in the weak sense if there exists a probabilistic classical algorithm with runtime poly$(n, \frac{1}{\epsilon}, \log \frac{1}{1-p})$ which outputs a number $E$ satisfying (\ref{def_strong}) with probability $p$. Thus for polynomial accuracies and for success probabilities which are exponentially (in $n$) close to 1, the classical simulation runs in poly$(n)$ time.

\

\noindent The motivation for the definition of a weak simulation originates from the fact the polynomial error scaling $\epsilon = 1/$poly$(n)$  captures how accurately the expectation value $\langle Z_1\rangle$ can be estimated by running the quantum circuit ${\cal C}_n$ polynomially many times. See section \ref{sect_chernoff} below and \cite{VandenNest2010Simulating} for a more extensive discussion.

The above definitions can readily be generalized to take into account more general inputs (e.g. arbitrary product states) and measurements (e.g. arbitrary single-qubit observables) as well as $d$-level systems. Finally, note that we will often use the term ``simulation'' as shorthand for ``efficient classical simulation''.

\subsection{Chernoff-Hoeffding bound}\label{sect_chernoff}

The Chernoff-Hoeffding bound is a tool to bound how accurately the expectation value of a random variable may be approximated using of ``sample averages''. Let $X_1, \dots X_K$ be i.i.d. real-valued random variables with $E := \mathbb{E}X_i$ and $X_i\in [-1, 1]$ for every $i=1, \dots, K$. Then  the Chernoff-Hoeffding bound asserts that \be\label{hoeff} \mbox{Prob} \left\{ \left|\frac{1}{K}\sum_{i=1}^K X_i - E\right| \leq \epsilon\right\} \geq 1-2 e^{- \frac{K\epsilon^2}{4}}.\ee For complex-valued  $X_i$ a similar bound can be
obtained for $|X_i|\leq 1$.

As an illustration, consider an $n$-qubit quantum circuit family ${\cal C}_n$ followed by measurement of $Z_1$ as in section \ref{sect_class_simulation}. Suppose that the circuit is run $K$ times, yielding an outcome $z_i\in\{1, -1\}$ in each run. Using (\ref{hoeff}) one shows that the number $E:= [\sum z_i]/K$, where the sum is over all $i=1\cdots K$,  satisfies $|E- \langle Z_1\rangle|\leq \epsilon $ with probability $p\geq 1-2 e^{- K\epsilon^2/4}$. Consequently, for any $\epsilon = 1/$poly$(n)$ there exists a suitable  $K=$ poly$(n)$ such that $|E-\langle Z_1\rangle|\leq \epsilon$ holds with probability  $p$ exponentially close to 1. In other words, the above procedure allows to achieve a polynomial approximation of $\langle Z_1\rangle$ in polynomial time with exponentially small probability of failure. This performance of the quantum computation corresponds precisely to the performance required of weak classical simulations, cf. section \ref{sect_class_simulation}.

\section{$2$-Local commuting circuits are easy} \label{sec:commutegatetwo}

Here we consider 2-local commuting circuits acting on general $d$-level systems. The main conclusion of this section will be that such circuits, when followed by single-qudit measurements, cannot outperform classical computation. In fact we will show that their power is even \emph{strictly} contained in P and give a concrete example of a simple function which cannot be computed with such commuting circuits.

\subsection{Efficient strong simulation of one qudit}

\begin{theorem}{\bf (Strong simulations of 2-local commuting circuits)}\label{thm_simulation_two_local}
Let ${\cal C}$ be a uniform family of 2-local $n$-qudit commuting circuits, acting on a product input state and followed by measurement of an observable $O$ acting on qudit $i$ for some $i$. Any such computation can be efficiently simulated classically in the strong sense.
\end{theorem}

\begin{proof}
We prove the result for $i=1$; other $i$ are treated fully analogously.  Denote the input by $|\alpha\rangle=|\alpha_1\rangle\cdots|\alpha_n\rangle$ where each $|\alpha_i\rangle$ is a single-qudit state. We can assume without loss of generality  that ${\cal C} = \prod U_{jk}$ is in standard form, where $U_{jk}$ represents the unique gate in the circuit with support $S\subseteq \{j, k\}$, for every $j, k =1\cdots n$ and $j<k$.
If $U_{jk}$ does not act on qudit 1, then this gate commutes with $O$. Hence in the product ${\cal C}^\dagger O{\cal C}$ we can commute $U_{jk}$ through ${\cal C}$ and $O$  to the left until it cancels out with $U_{jk}^\dagger$. By doing so, we can remove all gates that do not act on qudit 1. Therefore the expectation value of $O$ is
\begin{equation}
\langle O\rangle = \qsl{\alpha}{\cal C}^\dagger O{\cal C}\qs{\alpha} = \qsl{\alpha}(\prod U_{1j})^\dagger O\prod U_{1j}\qs{\alpha}
\end{equation}
where the products are over all $j\geq 2$. Now our strategy will be to trace out qudits one by one in the above equation.
Denote $|\alpha^{(1)}\rangle = |\alpha\rangle$, ${\cal C}^{(1)}={\cal C}$ and  $O^{(1)}=O$. Furthermore for every $k=2, \dots, n-1$ define \be |\alpha^{(k)}\rangle&=& \qs{\alpha_1}\qs{ \alpha_{k+1}}\cdots \qs{\alpha_n}\nonumber\\ {\cal C}^{(k)} &=& \ U_{1 k+1}\cdots U_{1n}\nonumber\\ O^{(k)}&=&[I\otimes \qsl{\alpha_k}]\ U_{1k}^{\dagger} O^{(k-1)} U_{1k}\ [ I\otimes \qs{\alpha_k}].\ee  Remark that each $O^{(k)}$ acts on a single qudit (namely qudit 1). Furthermore each of these operators can be computed classically with exponential precision in polynomial time: $O^{(1)}$ is given as an input and each update from $O^{(j)}$ to $O^{(j+1)}$ involves simple multiplications of 2-qudit operations which can be done in constant time (taking $O(d^6)$ steps where $d$ denotes the dimension of one qudit).

With the above definitions one finds, for every $k=2, \dots, n-1$:
\begin{equation}
\qsl{\alpha^{(k-1)}}[{\cal C}^{(k-1)}]^{\dagger}O^{(k-1)}{\cal C}^{(k-1)}\qs{\alpha^{(k-1)}}= \qsl{\alpha^{(k)}}[{\cal C}^{(k)}]^{\dagger}O^{(k)}{\cal C}^{(k)}\qs{\alpha^{(k)}}.
\end{equation}
Using this equation iteratively, we get
\begin{equation}
\langle O\rangle = \langle \alpha^{(1)}|[{\cal C}^{(1)}]^{\dagger} O^{(1)}{\cal C}^{(1)}|\alpha^{(1)}\rangle = \cdots
=\qsl{\alpha^{(n-1)}}U_{1n}^{\dagger} O^{(n-1)} U_{1n}\qs{\alpha^{(n)}}.
\end{equation}
The last expression is easily computed since $\qs{\alpha^{(n-1)}}$ is a 2-qudit state and $U_{1n}$ and $ O^{(n-1)}$ act on at most 2 qudits.
\end{proof}

The above result can readily be generalized in different ways. First, using a similar argument one shows that  measurement of any observable acting  on $O(\log n)$ qudits can be strongly simulated as well. Furthermore, interestingly, the result also generalizes to mutually \emph{anticommuting} gates, and more generally to gates which commute ``up to a phase'' as follows. Let ${\cal C}=\prod G_i$ be a uniform family of 2-local $n$-qudit  circuits such that $G_iG_j = \gamma_{ij}G_jG_i$ for all pairs of gates, where the $\gamma_{ij}$ are complex phases. Input and measurement are as in theorem \ref{thm_simulation_two_local}. Then such circuits can be efficiently simulated classically in the strong sense. Analogous to the first step in the proof of theorem \ref{thm_simulation_two_local}, the proof starts by ``removing'' all gates which do not act on qudit $i$ from the product ${\cal C}^{\dagger} O{\cal C}$  by commuting them through the circuit. This introduces an (easily computed) product of phases $\gamma_{ij}$. The remainder of the proof of theorem \ref{thm_simulation_two_local} carries over straightforwardly.

\subsection{2-local commuting circuits cannot compute all functions in $P$} \label{subsect:limited}

Here we show that two-local commuting circuits are not universal for classical computation by giving an explicit example of a function which is not computable with such circuits.

For every $d$ we let $\mathbf{Z}_d$ denote the set of integers modulo $d$. 
Let ${\cal C}$ denote a two-local commuting circuit acting on $m$ $d$-level systems. Consider a function $f:\mathbf{Z}_d^k\to \mathbf{Z}_d$. We say that ${\cal C}$ computes $f$ with probability at least $p$ if the circuit ${\cal C}$ acing on $|x, 0\rangle$ (where $0$ denotes a string of $m-k$ zeroes) and followed by a standard basis measurement of the first qudit yields the outcome $f(x)$ with probability at least $p$. 

We will in particular consider  the ``inner product function'' $f_{\mbox{\scriptsize{ip}}}:\mathbf{Z}_d^{2n}\to \mathbf{Z}_d$ defined by \be f_{\mbox{\scriptsize{ip}}}(x^a, x^b) = (x^a)^Tx^b \mod d,\ee  for every $x^a, x^b\in\mathbf{Z}_d^n$.

\begin{lemma}\label{lemma:compactset}
Let $\sigma_1,\dots \sigma_N$ be a collection of $d\times d$ density operators. For any $\epsilon>0$, if $N>\left(\frac{5}{\epsilon}\right)^{2d^2}$, then there exists two operators $\sigma_j$ and $\sigma_k$ such that  $\| \sigma_j-\sigma_k\|_{\mbox{\scriptsize{tr}}}\leq \epsilon$, where $\|A\|_{\mbox{\scriptsize{tr}}}\equiv \frac{1}{2}tr\sqrt{A^{\dag} A}$ denotes the trace distance.
\end{lemma}

\begin{proof}
We will show for any $\epsilon > 0$, there exists a finite set $\mathbf{E}$ of $d\times d$ density operators, such that for every density operator $\rho$, there exists $\sigma \in \mathbf{E}$ with $\|\rho-\sigma \|_{\mbox{\scriptsize{tr}}}< \epsilon $ (we call $\mathbf{E}$ a $\epsilon$-net). To do this, first we recall that every density operator of dimension $d$ has a purification by introducing an ancillary $d$-dimensional space $R$ . And in \cite{Hayden2004randomizing}, it was shown that for pure states of dimension $d^2$, there exists a $\epsilon$-net $\mathbf{F}$ with cardinality $|\mathbf{F}|\leq \left(\frac{5}{2\epsilon}\right)^{2d^2}\equiv M$. We can then choose set $\mathbf{E}$ to be $tr_R{\mathbf{F}}$, which is the partial trace of each element of set $\mathbf{F}$. Since partial trace is a contractive operation \cite{Nielsen2000quantum}, i.e. $\|tr_R(\mu - \tau)\|_{\mbox{\scriptsize{tr}}}\leq \|\mu - \tau\|_{\mbox{\scriptsize{tr}}}$, we know that set $\mathbf{E}$ obtained this way is indeed an $\epsilon$-net.

Note that $|\mathbf{E}|=|\mathbf{F}|=M$. Now if there are more then $M$ density operators, then there must be two density operators $\sigma_j$, $\sigma_k$ that are $\epsilon$-close to the same element of $\mathbf{E}$. Thus by triangle inequality $\|\sigma_j-\sigma_k\|_{\mbox{\scriptsize{tr}}}<2\epsilon$. The proof can be finished by a rescaling of $\epsilon$.
\end{proof}

\begin{theorem}
Consider an arbitrary $d$ and an arbitrary constant $p>1/2$. For sufficiently large $n$, the inner product function $f_{\mbox{\scriptsize{ip}}}$ is not computable by any two-local commuting circuit.
\end{theorem}

\begin{proof} Suppose there exists an $m$-qudit quantum circuit ${\cal C}$, for some $m\geq 2n$, which computes $f$ with probability $p>1/2$. We show that this leads to a contradiction.  Repeating the argument of theorem \ref{thm_simulation_two_local} we can remove all gates from the circuit which do not act on qudit 1. We  denote this simplified circuit again by ${\cal C}$. Now write ${\cal C}= {\cal C}_b {\cal C}_a$, where ${\cal C}_a$ consists of all gates in the circuit acting on qudits $\{1, i\}$ with $i=1\dots n$ and where ${\cal C}_b$ consists of all gates  acting on qudits $\{1, j\}$ with $i=n+1\dots m$. Furthermore, let $x=(x^a, x^b)$ be an arbitrary input of $f$. Finally, denote \be \sigma(x^a):= \mbox{ Tr}_{n\dots 2} \ {\cal C}_a|x^a\rangle\langle x^a|{\cal C}_a^{\dagger},\ee which is the reduced density operator for qudit 1 of the state ${\cal C}_a|x^a\rangle$.

The final state of the entire circuit is ${\cal C}|x,0\rangle$ where $0$ denotes a string of $m-n$ zeroes. With the notations above, the reduced density operator of the first qudit is \be\label{rho_xa_xb} \rho(x^a, x^b) &:=& \mbox{ Tr}_{m\dots 2}\ {\cal C}|x,0\rangle\langle x,0|{\cal C}^{\dagger}= \mbox{ Tr}_{m\dots n+1}\mbox{ Tr}_{n\dots 2} \ {\cal C}_b{\cal C}_a|x,0\rangle\langle x,0|{\cal C}_a^{\dagger}{\cal C}_b^{\dagger} \nonumber\\
&=&  \mbox{ Tr}_{m\dots n+1}\ {\cal C}_b \left\{\sigma(x^a)\otimes |x^b, 0\rangle\langle x^b, 0|\right\} {\cal C}_b^{\dagger}\ee
We now use lemma \ref{lemma:compactset}. This implies for every $\epsilon>0$ there exists an $n$ sufficiently large and two $n$-tuples $x^a\neq y^a$ such that $\| \sigma(x^a)-\sigma(y^a)\|_{\mbox{\scriptsize{tr}}}\leq \epsilon$. Using (\ref{rho_xa_xb}) and the fact that the trace norm is contractive, it follows that $\| \rho(x^a, x^b)-\sigma(y^a, x^b)\|_{\mbox{\scriptsize{tr}}}\leq \epsilon$ for \emph{every} $n$-tuple $x^b$! This implies the following: if a standard basis measurement on $\rho(x^a, x^b)$ yield some outcome $u$ with probability $p(u)$, then standard basis measurement on  $\rho(y^a, x^b)$ will yield the same outcome with probability $q(u)$ where $|p(u)-q(u)|\leq \epsilon$.  Setting $\epsilon = p - \frac{1}{2}$ and using that ${\cal C}$ computes $f$ with probability at least $p$, it then follows that $f(x^a, x^b)= f(y^a, x^b)$ for all $x^b$. Using the definition of $f$, this straightforwardly implies that $x^a=y^a$, thus leading to a contradiction. \end{proof}

\section{$3$-Local commuting circuits are hard}

Next we show that strong simulations of 3-local commuting circuits are unlikely to exist. 

\begin{theorem}[{\bf Hardness of simulating 3-local commuting circuits}]\label{thm_hardness_3_local}
Let ${\cal C}$ be a uniform family of $n$-qubit 3-local commuting quantum circuits acting on the input $|0\rangle$ and followed by $Z$ measurement of the first qubit. If all such circuits could be efficiently simulated classically in the strong sense then every problem in $\#P$ has a polynomial time algorithm.
\end{theorem}
In other words, there is a drastic increase in complexity in the seemingly innocuous transition from 2-local to 3-local gates. Remark that hardness already holds for the simplest case i.e. qubit systems---even though $d$-level 2-local commuting circuits have efficient simulations for any $d$. Hardness of strong simulations does not necessarily imply that weak simulations are hard as well since strong and weak simulations are generally inequivalent concepts (cf. \cite{VandenNest2010Simulating} for a discussion). In section \ref{sect_power_of_commuting} we will provide evidence that $k$-local commuting circuits with constant $k$  can efficiently perform certain tasks that appear to be nontrivial for classical computers, thereby providing evidence that efficient weak simulations might not exist in general.

The proof of theorem \ref{thm_hardness_3_local} is given below. Our approach is to relate simulations of 3-local
commuting circuits to the evaluation of matrix elements of \emph{universal}  unitary
quantum circuits, which is known to be hard. The following three lemmata collect preliminary results.  First we recall that  the evaluation of  matrix elements of  universal quantum circuits is known to be hard. We denote $S:=$ diag$(1, e^{i\pi/4})$ and $CZ:=$ diag$(1, 1, 1, -1)$.

\begin{lemma}\label{thm_hardness_ME}
Let ${\cal U}$ be a uniform family of $n$-qubit quantum circuits  composed of the gates $H$, $S$ and $CZ$. If there existed an algorithm with runtime poly$(n, \log \frac{1}{\epsilon})$  which  outputs an $\epsilon$-approximation of $\langle 0| {\cal U}|0\rangle$  for any such circuit family,  then every problem in $\#P$ has a polynomial-time algorithm.
\end{lemma}
\begin{proof} Consider an efficiently computable Boolean function $f:\{0, 1\}^n\to \{0, 1\}$. Let $s(f)$ denote the number of bit strings $x$ satisfying $f(x)=0$. The problem of computing $s(f)$ is well known to $\#$P-complete. Now define the $(n+1)$-qubit state $|f\rangle:= 2^{-n/2}\sum_x |x, f(x)\rangle$ where the sum is over all $n$-bit strings $x$.  Let ${\cal H}$ be the operator which acts as $H$ on qubits 1 to $n$ and as the identity on qubit $n+1$. Then an easy calculation shows \be\label{s_f} \langle 0| {\cal H}|f\rangle = s(f)/2^n.\ee
Since $H$, $CZ$ and $S$ form a universal gate set,  the Solovay-Kitaev theorem  implies that there  exists a uniform circuit family ${\cal V}$  composed of these gates such that ${\cal V}|0\rangle$ is $\delta$-close to $|f\rangle$ with $\delta:= 2^{-n^2}$. Denote the circuit ${\cal U}:={\cal H}{\cal V}$.  Using (\ref{s_f}) it  follows that \be\label{sharp_P_universal} |\langle 0| {\cal U}|0\rangle - \frac{s(f)}{2^n}|\leq \delta.\ee
Now suppose that there exists a poly$(n, \log\frac{1}{\epsilon})$ classical algorithm to compute $\langle 0| {\cal U}|0\rangle$ with accuracy $\epsilon$. Setting $\epsilon = \delta$, this would imply the existence of a polynomial time classical algorithm that outputs an $\delta$-approximation $\gamma$ of $\langle 0| {\cal U}|0\rangle$. Using (\ref{sharp_P_universal}) and the triangle inequality this implies that $\gamma$ approximates $s(f)/2^n$ with accuracy $2\delta$. Since $s(f)/2^n = k/2^n$ for some integer between 0 and $2^n$,   this accuracy would allow to compute $s(f)$ exactly in polynomial time, hence implying that  every problem in $\#$P has a poly-time algorithm. 
\end{proof}

Second, we recall a result from \cite{Bremner2010Classical} which relates universal quantum circuits to post-selected 2-local commuting circuits.

\begin{lemma}\label{thm_postselect}
Let ${\cal U}$ be an $n$-qubit quantum circuit  composed of the gates $H$, $S$ and $CZ$ and denote $|\psi\rangle = {\cal U}|0\rangle^n$. Then there exists a 2-local commuting circuit ${\cal C}$ on $k+n$ qubits such that $|\psi\rangle$ is obtained by postselecting ${\cal C}|0\rangle^{k+n}$ on the first $k$ qubits; more precisely \be\label{postselect} |0\rangle^{k} |\psi\rangle = \sqrt{2}^k{\cal P}{\cal C}|0\rangle^{k+n}.\ee Here ${\cal P}$ denotes the projector $|0\rangle\langle 0|$ acting on the first $k$ qubits. Furthermore $k=$ poly$(n)$ and the description of ${\cal C}$ can be computed efficiently on input of the description of ${\cal U}$.
\end{lemma}

Combining the above two lemmata shows that approximating matrix elements of commuting 2-local circuits is hard.
\begin{lemma}
Let ${\cal C}$ be a uniform family of $n$-qubit 2-local commuting quantum circuits. If there existed a classical algorithm with runtime poly$(n,\log \frac{1}{\epsilon})$  which outputs an $\epsilon$-approximation of $\langle 0| {\cal C}|0\rangle$  for any such ${\cal C}$,  then every problem in $\#P$ has a poly-time algorithm.
\end{lemma}
\begin{proof} Let ${\cal U}$ be a uniform family of $n$-qubit quantum circuits composed of the gates $H$, $S$ and $CZ$ and  let ${\cal C}$ be the associated commuting circuit family as in lemma \ref{thm_postselect}. Using (\ref{postselect}) one finds \be \langle 0|^n{\cal U}|0\rangle^n = \sqrt{2}^k \langle 0|^{n+k}{\cal C}|0\rangle^{n+k}.\ee If an efficient classical algorithm existed to estimate $\langle 0| {\cal C}|0\rangle$ with exponential precision, then there also exists an algorithm to estimate $\langle 0|{\cal U}|0\rangle $ with exponential precision. This implies that every problem in $\#P$ has a poly-time algorithm owing to lemma \ref{thm_hardness_ME}.
\end{proof}

The proof of theorem \ref{thm_hardness_3_local} now proceeds by relating the simulation of 3-local commuting circuits to the evaluation of matrix elements of 2-local commuting circuits, via the Hadamard test.

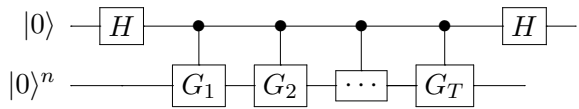
\begin{figure}
\[
\Qcircuit @C=1em @R=.7em {
\lstick{\qs{0}} & \gate{H} & \ctrl{1} & \ctrl{1} & \ctrl{1} &\ctrl{1}& \gate{H}& \qw \\
\lstick{\qs{0}^n} & \qw & \gate{G_1} & \gate{G_2}& \gate{\cdots}& \gate{G_T}\qw & \qw
}
\]
\caption[]{The Hadamard test}\label{fig_hadamard}
\end{figure}

{\bf Proof of theorem \ref{thm_hardness_3_local}:} Suppose that an efficient algorithm existed to strongly simulate the circuits described in the theorem.  Consider an arbitrary $n$-qubit  commuting circuit ${\cal C}= G_T\cdots G_1$ with two-qubit gates $G_i$. Consider the following $(n+1)$-qubit quantum circuit (the ``Hadamard test'') with input  $|0\rangle$ as depicted in Fig. \ref{fig_hadamard}. First $H$ is applied to the first qubit. Then each gate $G_i$ is applied controlled on the first qubit being in the state $|1\rangle$; we denote these 3-qubit gates by $CG_i$. Finally, $H$ is again applied to the first qubit. Measuring the first qubit yields the outcome 0 with probability \be p(0) = \frac{1}{2}(1 + \mbox{ Re}(\langle 0|{\cal C}|0\rangle)).\ee Now for each $i$ define the 3-qubit gate $U_i:= [H\otimes I] CG_i [H\otimes I]$, where $H$ acts on the first qubit, and let ${\cal C}'$ denote the circuit composed of the gates $U_i$. Since the gates $G_i$ commute, also the gates $U_i$ commute. Furthermore, it is straightforward to show that the circuit ${\cal C}'$ acting on $|0\rangle$ and followed by measurement of the first qubit is equivalent to the circuit in Fig. \ref{fig_hadamard},  since the hadamard operations ``in the middle'' cancel out. Thus ${\cal C}'$ also yields the outcome $0$ with probability $p(0)$. It follows that the existence of an efficient classical algorithm to strongly simulate the circuit ${\cal C}'$ yields an efficient classical algorithm to compute the real part of $\langle 0|{\cal C}|0\rangle$ with exponential precision. Replacing the second Hadamard gate in Fig. \ref{fig_hadamard} by $PH$ where $P=$ diag$(1, i)$ and arguing analogously yields an efficient algorithm to estimate  the imaginary part of $\langle 0|{\cal C}|0\rangle$ with exponential precision. Using lemma \ref{thm_postselect} we conclude that  this would imply  that every problem in $\#P$ has a poly-time algorithm.
\finpr

\section{Efficient simulation of commuting Pauli Circuits}\label{sect_pauli}
A circuit composed of unitary operators of the form $e^{i\theta P}$, where the $P$s are (Hermitian) Pauli operators, is called a Pauli circuit.  Recall that every two Pauli operators either commute or anticommute. Pauli circuits are easily seen to be universal for quantum computation. Here we investigate \emph{commuting} Pauli circuits. We allow  $P$ to act on arbitrarily many qubits i.e. we do not restrict to local gates\footnote{Remark that, even for such non-local gates, every gate $e^{i\theta P}$ can be efficiently implemented on a quantum computer i.e.  it can be realized by a polynomial size quantum circuit of elementary gates}.

Given the distinguished status of Pauli operators, commuting Pauli circuits constitute  a simple and natural class of commuting quantum circuits. This class in fact encompasses the model of ``instantaneous quantum computation'' (IQP) introduced in \cite{Shepherd2009instantaneous}. IQP  corresponds to the subclass of commuting Pauli circuits where each $P$ is restricted to be a tensor product of identities and Pauli $X$ matrices, so that every gate $e^{i\theta P}$ is diagonalized by the tensor product operator $H\otimes\cdots \otimes H$. Generalizing IQP to arbitrary commuting Pauli circuits adds the interesting feature that the unitary operator which simultaneously diagonalizes  the gates in the circuit is generally no longer a tensor product of single-qubit operators, but rather a global entangling operation; see example (\ref{K_k}).

Whereas arbitrary Pauli circuits are universal, we will show that commuting Pauli circuits can be efficiently simulated classically in the following sense.

\begin{theorem}{\bf (Weak simulation of Commuting Pauli circuits)}\label{thm_Pauli_simulation}
Every uniform family of commuting Pauli circuits acting on a standard basis input and followed by  measurement of $Z$ acting on one of the qubits can be weakly simulated classically.
\end{theorem}
It was shown in \cite{Bremner2010Classical} that  IQP circuits followed by single-qubit standard basis measurements can be simulated efficiently weakly\footnote{In fact classical simulations were also achieved in \cite{Bremner2010Classical} for $O(\log n)$ measurements; our results can also be generalized to simulate such measurements for arbitrary commuting Pauli circuits.}. Theorem \ref{thm_Pauli_simulation}  hence generalizes this result to arbitrary commuting Pauli circuits.
Furthermore, in \cite{Bremner2010Classical}  it was shown that efficient weak classical simulation (relative to a certain special type of approximations viz. multiplicative approximations) of  2-local IQP circuits followed by $O(n)$ computational basis measurements are highly unlikely to exist: the existence of such simulations would imply a collapse of the polynomial hierarchy to its third level. Thus a fortiori simulations of $O(n)$ computational basis measurements are unlikely to exist for general commuting Pauli circuits as well.

One can in fact show a stronger version of theorem \ref{thm_Pauli_simulation}: a general Pauli circuit containing a limited degree of non-commutativity can still be simulated classically efficiently.

\begin{theorem}{\bf (Weak simulation of slightly non-commuting Pauli circuits)}\label{thm_Pauli_non-communing}
Consider a uniform family of $n$-qubit  commuting Pauli circuits interspersed with $O(\log n)$  gates of the form $e^{i\theta Q}$ with $Q$ an arbitrary (Hermitian) Pauli operator. Any such circuit family acting on standard basis input and followed by  measurement of $Z$ acting on one of the qubits can be weakly simulated classically.
\end{theorem}

The proofs of theorem \ref{thm_Pauli_simulation} and \ref{thm_Pauli_non-communing} are given in section \ref{sect_proof_Pauli_simulation}. In the preceding sections we develop the necessary tools. It is interesting that the simulation techniques used here are completely different from those used our simulations of 2-local commuting circuits (theorem \ref{thm_simulation_two_local}). In particular the latter involved strong simulations whereas commuting Pauli circuits will be simulated using weak simulations combined with stabilizer methods.

\subsection{Pauli and Clifford operators}

A  Pauli operator on $n$ qubits has the form $P= \alpha P_1\otimes \ldots \otimes P_n$, where $\alpha\in\{\pm 1, \pm i\}$ and where each $P_j$ is one of the Pauli matrices $X$, $Y$, $Z$ or the identity. A Pauli operator is said to be of $Z$-type if each $P_j$ is either $Z$ or the identity; $X$-type Pauli operators are defined analogously.  Since $X$, $Y$ and $Z$ are Hermitian, a Pauli operator is Hermitian if and only if $\alpha\in\{1, -1\}$.  Letting $Z_k$ and $X_k$ denote the operators $Z$ and $X$ acting on qubit $k$, respectively, it can be verified that every Pauli operator $P$ can be written as
\be\label{Pauli}
P=i^{t}\prod_k X_k^{a_k}Z_k^{b_k},\quad \mbox{ where } t \in \{0,1,2,3\},\ a_k, b_k \in \{0,1\}.
\ee
Defining the $2n$-dimensional bit string \be r(P) = (a_1,\cdots, a_n, b_1,\cdots, b_n),\ee it is easily verified that $r(PQ)= r(P)+r(Q)$ for all Pauli operators $P$ and $Q$, where addition is modulo 2.

An $n$-qubit operator $U$ is a Clifford operation if $UPU^{\dagger}$ is a Pauli operator for every Pauli operator $P$. The set of all $n$-qubit Clifford operations is a group, called the Clifford group. A Clifford circuit is a quantum circuit composed of $H$, $CNOT$ and $P=$ diag$(1, i)$. It is well known that every Clifford circuit realizes a Clifford operator, and that every Clifford operator can  be realized as a (polynomial-size) Clifford circuit.

\begin{lemma}\label{commuting_pauli}
Let $P_1, \dots, P_m$ be a collection of commuting $n$-qubit Pauli operators. Then there exists a Clifford operation ${\cal C}$ such that ${\cal C}^{\dagger}P_i{\cal C} = Q_i$ for every $i$, where each $Q_i$ is a  $Z$-type Pauli operator. Moreover each $Q_j$ as well as the description of a poly-size Clifford circuit realizing ${\cal C}$ can be determined efficiently.
\end{lemma}
\begin{proof} It suffices to prove the result for Hermitian Pauli operators since every Pauli operator can be made Hermitian by providing it with a suitable overall phase. Thus henceforth we assume that the $P_i$ are Hermitian. We can write all $m$ vectors $r(P_i)$ in a $m\times 2n$ matrix and pick out a maximal set of independent row vectors over $\mathbb{Z}_2$ efficiently by Gaussian elimination. W.l.o.g. we assume these are the first $l$ vectors. The corresponding Pauli operators $\{P_1, \dots, P_l\}=:{\cal S}$ form an independent set i.e. no operator in ${\cal S}$ can be written as a product of the other elements of ${\cal S}$. In addition, no product of operators in ${\cal S}$ yields $-I$. Indeed suppose there exist bits $x_j$, not all zero,  such that $P_1^{x_1}\dots P_l^{x_l} = -I$. This would imply that $\sum x_j r(P_j)=0$, contradicting with the linear independence of the $r(P_j)$. Since the operators in ${\cal S}$ are Hermitian, independent and commuting and since no product of some of these operators yields $-I$, there exists a stabilizer code ${\cal V}$ of dimension $2^{n-l}$ stabilized by ${\cal S}$ \cite{Nielsen2000quantum}. This implies in particular that $l\leq n$. Using standard stabilizer techniques one can efficiently compute additional Hermitian Pauli operators ${\cal S}'=\{R_{l+1}, \dots, R_n\}$ such that all operators in the set  ${\cal T}= {\cal S}\cup {\cal S}'$ mutually commute, are independent and no product of these operators yields $-I$ \cite{Nielsen2000quantum}. These $n$ operators are the stabilizers of a 1-dimensional stabilizer code i.e. a stabilizer state $|\psi\rangle$. In other words $|\psi\rangle$ satisfies $P_i|\psi\rangle= |\psi\rangle = R_j|\psi\rangle$ for every $i=1, \dots, l$ and $j=l+1, \dots, n$, and moreover it is the unique state doing so. It is well known that there exists a poly-size $n$-qubit Clifford circuit ${\cal C}$ such that $|\psi\rangle = \gamma {\cal C}|0\rangle^n$ for some global phase $\gamma$; moreover a description of ${\cal C}$ can be computed efficiently \cite{Dehaene2003the}.
Now define $Q_i = {\cal C}^{\dagger}P_i {\cal C}$ for every $i=1, \dots, m$. Each $Q_i$ is an efficiently computable Pauli operator since ${\cal C}$ is a poly-size Clifford circuit.  Since ${\cal C}|0\rangle = |\psi\rangle$ and $P_j|\psi\rangle = |\psi\rangle$ for every $P_j\in {\cal S}$ one has $Q_j|0\rangle = |0\rangle$. This last property together with the fact that each $Q_j$ is a Pauli operator  implies that $Q_j$ must be of $Z$-type. Finally, since each $P_k$ with $k\geq l+1$ can be written, up to a global phase,  as a product of operators within ${\cal S}$ and since products of $Z$-type Pauli operators are again of $Z$-type,  it follows that also $Q_k$ is of $Z$-type.
\end{proof}

\subsection{CT states}\label{sect_CT}

Here we recall a result (theorem \ref{thm_CT} below) stating that a general class of quantum processes can be simulated weakly. First we need some definitions. Consider a family of $n$-qubit states $|\psi_n\rangle \equiv |\psi\rangle$  specified in terms of some classical description, say a quantum circuit preparing $|\psi\rangle$ from the state $|0\rangle$. Following \cite{VandenNest2010Simulating}, $|\psi\rangle$ is said to be \emph{computationally tractable (CT)} (relative to this description) if
\begin{itemize}
\item[(a)] it is possible to sample in poly$(n)$ time with classical means from the probability distribution Prob$(x)=|\qinner{x}{\psi}|^2$ on the set of $n$-bit strings $x$, and
\item[(b)] for any bit string $x$, the coefficient $\qinner{x}{\psi}$ can be computed in poly$(n)$ time on a classical computer with exponential precision.
\end{itemize}
Second, an $n$-qubit unitary operator $U$ is said to be monomial if there exists a permutation $\pi$ on the set of $n$-bit strings and a family of complex phases $\lambda_x$, such that \be U|x\rangle = \lambda_x |\pi(x)\rangle\quad \mbox{ for every }|x\rangle.\ee In other words $U$ maps each standard basis state to another one, up to a global phase. Equivalently, one has $U=PD$ where $D=\sum\lambda_x|x\rangle\langle x|$ is a diagonal matrix and $P = \sum |\pi(x)\rangle\langle x|$ is a permutation matrix. The operation $U$ is said to be efficiently computable if the functions \be x\to \lambda_x, \quad x\to\pi(x)\quad \mbox{ and } \quad x\to \pi^{-1}(x)\ee can be computed efficiently.

In our simulation of commuting Pauli circuits we will use the following classical simulation result proved in \cite{VandenNest2010Simulating}.
\begin{theorem}{\bf (CT states)}\label{thm_CT}
Let $\qs{\psi}$ and $|\varphi\rangle$  be  $n$-qubit CT states and let $U$ be an $n$-qubit efficiently computable monomial operation. Then there exists a polynomial time  classical algorithm to approximate $\qsl{\psi}U\qs{\varphi}$ with polynomial accuracy (and exponentially small probability of failure).
\end{theorem}
For our purposes, it will be relevant that every stabilizer state is CT. More precisely,  for every polynomial-size Clifford circuit ${\cal C}$ and standard basis state $|x\rangle$ (where $x$ is an $n$-bit string), the state $|\psi\rangle = {\cal C}|x\rangle$ is CT relative to  the description of ${\cal C}$ and the input $x$. Property (a) is the content of the Gottesman-Knill theorem \cite{Gottesman1997Stabilizer}. Property (b) was shown in \cite{Dehaene2003the}; in fact for every stabilizer state $|\psi\rangle$ the standard basis coefficients $\langle y|\psi\rangle$ can be computed exactly. We refer to \cite{VandenNest2010Simulating} for a more extensive discussion of CT states.

As for monomial operators, it is easily shown using (\ref{Pauli}) that every Pauli operator is unitary, monomial and efficiently computable. Second,  every unitary operator of the form $\exp[i\theta Q]$, where $Q$ is any (Hermitian) $Z$-type Pauli operator, is diagonal and hence monomial. Furthermore it is straightforward to show that any such operator is efficiently computable. More generally, it is useful to note (and easy to show):
\begin{lemma}\label{thm_prod_monomial}
If $U_1, \dots, U_k$ are efficiently computable monomial unitary $n$-qubit operators and $k=$ poly$(n)$, then also $\prod_{i=1}^k U_i$ is efficiently computable monomial.
\end{lemma}

\subsection{Proof of theorem \ref{thm_Pauli_simulation}} \label{sect_proof_Pauli_simulation}

For clarity we prove theorem \ref{thm_Pauli_simulation} separately even though it is superseded by theorem \ref{thm_Pauli_non-communing}. Denote the input by $|x\rangle$ where $x$ is an $n$-bit string. Denote the Pauli circuit by ${\cal U}$ and let $e^{i\theta_j P_j}$ denote its gates $(1\leq j\leq m)$. Let $\langle Z_i\rangle$ denote the expectation value of $Z$. First we invoke lemma \ref{commuting_pauli}, yielding a Clifford circuit ${\cal C}$ satisfying ${\cal C}^{\dagger} P_j{\cal C} = Q_j$ for some efficiently computable Hermitian $Z$-type operators $Q_j$. It follows that \be e^{i\theta_j P_j} = {\cal C} e^{i\theta_j Q_j}{\cal C}^{\dagger}\ee and therefore ${\cal U} = {\cal C} {\cal D}{\cal C} ^{\dagger}$  where ${\cal D}$ is given by the product of the $m$ diagonal operators $e^{i\theta_j Q_j}$. Denote $P={\cal C}^{\dagger} Z_i {\cal C}$ which is an efficiently computable Pauli operator. Furthermore denote $|\psi\rangle:= {\cal C}^{\dagger}|x\rangle$. Then \be \langle Z_i\rangle = \langle x|{\cal U}^{\dagger} Z_i {\cal U}|x\rangle = \langle \psi|{\cal D}^{\dagger} P {\cal D}|\psi\rangle.\ee Since ${\cal C}$ is a Clifford circuit, $|\psi\rangle$ is a CT state. Finally $M:={\cal D}^{\dagger} P {\cal D}$ is monomial and efficiently computable: indeed the Pauli operator $P$ as well as each $e^{\theta_j Q_j}$ are efficiently computable monomial, as discussed in section \ref{sect_CT}. Applying lemma \ref{thm_prod_monomial} then shows that $M$ is efficiently computable monomial as well. Theorem \ref{thm_CT} can now be applied.

\subsection{Proof of theorem \ref{thm_Pauli_non-communing}}

We assume w.l.o.g. that $Z$ is measured on the first qubit. Let ${\cal C}'$ be obtained by interspersing the commuting Pauli circuit ${\cal C}= \prod e^{i\theta P_j}$ with  $k$ additional gates $e^{i\theta Q_{j}}$ at arbitrary places in the circuit.  Write \be e^{i\theta Q_j} = [\cos \theta] I + [i\sin \theta] Q_j\ee  for ever such additional gate. Doing so, the circuit ${\cal C}'$ is written as a linear combination of $2^k$ circuits (with coefficients of the form $(\cos \theta)^l(i\sin\theta)^{k-l}$), each of which being obtained by replacing $e^{i\theta Q_j}$ by either $I$ or $Q_j$. Thus every circuit in the linear combination is obtained by interspersing ${\cal C}$ with $k$ Pauli operators. Using that $e^{i\theta P} Q = Q e^{\pm i\theta P}$ for every two Pauli operators $P$ and $Q$,  the $Q_j$ can all be commuted to the right. As a result, we find that ${\cal C}'$ is written in the form \be {\cal C}' = \sum_{\alpha=1}^{2^k} a_{\alpha} {\cal C}_{\alpha}\Sigma_{\alpha} ,\ee where each coefficient $a_{\alpha}$ is efficiently computable, where each $\Sigma_{\alpha}$ is a Pauli operator and where each ${\cal C}_{\alpha}$ is a commuting Pauli circuit obtained by flipping a subset of the signs $P_j\to -P_j$ in the commuting circuit ${\cal C}$. Furthermore there are only poly$(n)$ terms in the sum since $k=O(\log n)$ by assumption. To arrive at an efficient weak simulation of ${\cal C}'$ followed by measurement of $Z_1$, it suffices to show that each of the matrix elements \be\label{Pauli_proof_matrix_element} \langle x|\Sigma_{\alpha} {\cal C}_{\alpha}^{\dagger} Z_1 {\cal C}_{\beta}\Sigma_{\beta}|x\rangle \ee can be estimated efficiently with polynomial accuracy. First we can commute $Z_1$ to the right,  transforming ${\cal C}_{\beta}$ into  a commuting Pauli circuit $\overline{{\cal C}}_{\beta}$ obtained by changing some of the signs $P_j\to \pm P_j$ as before. Note that the combined circuit ${\cal C}^{\dagger}_{\alpha} \overline{{\cal C}}_{\beta}$ is a commuting Pauli circuit since all gates have the form $e^{\pm i\theta_j P_j}$. Furthermore $\Sigma_{\alpha}|0\rangle$ and $Z_1\Sigma_{\beta}|0\rangle$ are, up to global phases, simple standard basis states,  say $|y\rangle$ and $|z\rangle$ resp., which can be computed efficiently. Analogous to the proof of theorem \ref{thm_Pauli_simulation} we write ${{\overline{\cal C}}^{\dagger}_{\alpha}}{\cal C}_{\beta} = {\cal U}{\cal D}{\cal U}^{\dagger}$ where ${\cal U}$ is a polynomial size Clifford circuit and ${\cal D}$ is a product of diagonal gates. Putting everything together we find that (\ref{Pauli_proof_matrix_element}) is, up to an efficiently computable overall phase, of the form $\langle y|{\cal U}{\cal D}{\cal U}^{\dagger}|z\rangle$ for some  standard basis states $|y\rangle$ and $|z\rangle$. Since ${\cal U}^{\dagger}|y\rangle$ and ${\cal U}^{\dagger}|z\rangle$ are CT states (see section \ref{sect_CT}) and since ${\cal D}$ is efficiently computable monomial, we can apply theorem \ref{thm_CT} yielding an efficient classical algorithm to estimate (\ref{Pauli_proof_matrix_element}). This proves the result.

\section{Mapping non-commuting circuits to commuting circuits}\label{sect_power_of_commuting}

Here we show that commuting circuits can be used to efficiently reproduce the output of certain non-commutative processes. These results will provide evidence that commuting circuits can be used to solve tasks that appear nontrivial for classical computers.

\subsection{Two-layer circuits}

For every constant $k$ we let $\Gamma^k$ denote a computational model involving a universal classical computer supplemented with a restricted quantum computer  operating with uniformly generated families of $k$-local commuting circuits acting on an arbitrary product input state and followed by $Z$ measurement of the first qubit. By construction,  $\Gamma^k$ has the power to efficiently solve every problem in the complexity class P, for every $k$. Our goal is to investigate whether $\Gamma^k$-computations have the potential to outperform classical computers.

\begin{theorem}{\bf(Mapping $k$-local non-commuting to $(k+1)$-local commuting circuits)}\label{thm_comm_non_comm}
Let ${\cal C}_1$ and ${\cal C}_2$ be uniform families of $k$-local $n$-qubit commuting circuits, where the gates in ${\cal C}_1$ need not commute with those in ${\cal C}_2$. Then there exists a polynomial time $\Gamma^{k+1}$-algorithm which approximates  $\langle 0|{\cal C}_1{\cal C}_2|0\rangle$ with polynomial accuracy  (with success probability exponentially close to 1).
\end{theorem}
The above result shows that the non-commutativity in the two-layer circuit ${\cal C}_1{\cal C}_2$ can be ``removed'' by allowing gates to act on $k+1$ qubits. The proof is an immediate consequence of the following alternate version of the Hadamard test (which regards arbitrary, i.e. not necessarily commuting, circuits).

\begin{lemma}{\bf (Alternate Hadamard test)}\label{thm_improved_hadamard}
Let ${\cal U} = U_{2m}\cdots U_1$ be an $n$-qubit quantum circuit of even size $2m$. Add one extra qubit line (henceforth called qubit 1) and for every $i=1 \cdots m$ define the gate \be\label{W_gates} W_i = |0\rangle\langle 0| \otimes U^{\dagger}_{2m+1-i} + |1\rangle\langle 1|\otimes U_{i},\ee which acts on qubit 1 and the qubits on which $U_i$ and $U_{i+k}$ acted in the initial circuit ${\cal U}$.  Consider the following circuit ${\cal U}'$ acting  on the $(n+1)$-qubit  input $|0\rangle$: first, apply $H$ to qubit 1; second, apply the gates $W_1, \dots, W_m$; third, apply $H$ to qubit 1; finally measure $Z$ on qubit 1. Then the probability of outputting 0 is  \be p(0) = \frac{1}{2}(1 + \mbox{Re}\langle 0|{\cal U}|0\rangle).\ee Analogously, replacing $H$ in the third step by $HP$ with $P=$ diag$(1, i)$ yields the imaginary part of $\langle 0|{\cal U}|0\rangle$.
\end{lemma}
Remark that lemma \ref{thm_improved_hadamard} requires ${\cal U}$ to have even size. This is however not an essential requirement since a circuit of odd size $2m+1$ can be ``padded'' with an additional identity. This yields a circuit ${\cal U}'$ of size $m+1$.

The proof of the lemma is obtained by directly computing $p(0)$. Similar to the Hadamard test, the above result provides a simple quantum algorithm to estimate matrix elements of unitary quantum circuits with polynomial accuracy (and with success probability exponentially close to 1).
Different from the standard Hadamard test, however, is that the \emph{size} of the circuit ${\cal U}'$ used in lemma \ref{thm_improved_hadamard} is  \emph{half} the size of the original circuit ${\cal U}$ i.e. the alternate Hadamard test is ``twice as fast''. The price to pay for this is that the gates in ${\cal U}'$ act on a larger number of qubits: if ${\cal U}$ is a $k$-local circuit then ${\cal U}'$ can be as much as  $(2k+1)$-local.

\

{\bf Proof of theorem \ref{thm_comm_non_comm}:} Without loss of generality we can assume that ${\cal C}_1$ and ${\cal C}_2$ are in standard form, say ${\cal C}_1= G_m\cdots G_1$ and ${\cal C}_2= G_m'\cdots G_1'$ where $m = {n \choose k}$. By definition of the standard form, for every subset $S$ of $k$ qubits there is precisely one gate $G_i$ and one gate $G_j'$ such that supp$(G_i)\subseteq S$ and supp$(G_j')\subseteq S$. By suitably labeling the gates in both circuits  we can ensure that always $j = m+1-i$. Now apply lemma \ref{thm_improved_hadamard} to the circuit ${\cal U}:={\cal C}_1{\cal C}_2$ with the identification $U_i := G_i$ and $U_{m+i} := G_i$ for every $i=1\cdots m$. Then  each gate (\ref{W_gates}) acts on the qubits in $S$ together with qubit $1$ so that this gate is  $(k+1)$-local (at most). Note furthermore that all gates $W_i$ mutually commute. Finally, define $W_i':= [H\otimes I] W_i[H\otimes I]$ where $H$ acts on qubit 1. Since all $H$ gates in the middle cancel out, the $(k+1)$-local commuting circuit ${\cal C} = \prod_i  W_i'$ acting on $|0\rangle$ followed by measurement of $Z_1$ yields the same output as the circuit ${\cal U}'$ of lemma \ref{thm_improved_hadamard}. This allows to estimate the real part of $\langle 0|{\cal C}_1{\cal C}_2|0\rangle$ with polynomial accuracy within the class $\Gamma^{k+1}$. The imaginary part is treated analogously. \finpr

\subsection{Constant-depth circuits}

Here we will relate commuting circuits with constant-depth circuits comprising \emph{arbitrary} gates.

\begin{theorem}\label{thm_bounded_depth}{\bf (Estimating constant-depth matrix elements)}
Let ${\cal U}$ be a uniform family of  $n$-qubit quantum circuits of constant depth $m$.  Then there exists a polynomial time $\Gamma^{k}$-algorithm  to approximate $|\langle 0|{\cal U}|0\rangle|^2$  with polynomial accuracy (and with success probability exponentially  close to 1) where $k=2^m+1$.
\end{theorem}
Recall that the problem of estimating matrix elements $|\langle 0|{\cal U}|0\rangle|$ of polynomial size quantum circuits of arbitrary depth is known to be BQP-hard (and the naturally corresponding decision problem is BQP-complete). Theorem \ref{thm_bounded_depth} shows that such matrix elements can be estimated efficiently with $k$-local commuting circuits with constant $k$ as long as ${\cal U}$ has constant depth (with an exponential scaling of $k$ with $m$). Although one would not expect the constant-depth matrix problem to be BQP-hard, this task appears to be nontrivial for classical computers and, to our knowledge, no efficient classical algorithm is known.

\

{\bf Proof of theorem \ref{thm_bounded_depth}:} Letting $Z_j$ denote the operator $Z$ acting on qubit $j$, we define $Z(S)= \prod_{j\in S} Z_j$ for every subset $S\subseteq\{1, \dots, n\}$.  Using \be |0\rangle\langle 0| = \frac{1}{2^n}\sum Z(S),\ee where the sum is over all subsets $S$, one finds \be\label{ME_bounded_depth1} |\langle 0|{\cal U}|0\rangle|^2 = \langle 0| {\cal U}^{\dagger} |0\rangle\langle0|{\cal U}|0\rangle = \frac{1}{2^n} \sum \langle 0|{\cal U}^{\dagger}Z(S) {\cal U}|0\rangle.\ee  Setting $G_j:= {\cal U}^{\dagger}Z_j {\cal U}$ yields \be \label{ME_bounded_depth2}\langle 0|{\cal U}^{\dagger}Z(S) {\cal U}|0\rangle = \langle 0|\prod_{j\in S} G_j|0\rangle=: F(S)\ee for every subset $S$. Since the $Z_j$ mutually commute, the $G_j$ mutually commute as well as these operators are obtained by simultaneously conjugating the $Z_j$. Furthermore since ${\cal U}$ has depth $m$, each $G_j$ acts on at most $2^m$ qubits. Thus $F(S)$ is a matrix element of a $2^m$-local commuting circuit. Via the standard Hadamard test (recall Fig. \ref{fig_hadamard} and the proof of theorem \ref{thm_hardness_3_local}) one constructs a $k$-local commuting circuit with $k=2^m+1$ which allows to estimate any such matrix element with polynomial accuracy in polynomial time, with success probability exponentially close to 1.

We now use these findings to give an efficient $\Gamma^k$-algorithm  to estimate $\gamma:=|\langle 0|{\cal U}|0\rangle|^2$ with polynomial accuracy. Owing to (\ref{ME_bounded_depth1})-(\ref{ME_bounded_depth2}), one has $\gamma:= 2^{-n} \sum F(S)$. Thus $\gamma$ equals the expectation value of a random variable over the collection of all $2^n$ subsets $S$  which takes the value $F(S)$ with uniform probability.   Fix $\epsilon>0$. First we generate $K$ subsets $S_{\alpha}\subseteq\{1, \dots, n\}$ uniformly at random.  Applying the Chernoff-Hoeffding bound we find that, for some sufficiently large $K = $ poly$(n, 1/\epsilon)$, one has \be\label{chernoff_bounded_depth} \left|\frac{1}{K} \sum_{\alpha = 1}^K F(S_{\alpha}) - \gamma\right|\leq \epsilon/2\ee with probability exponentially close to 1. Next, as described above we can efficiently compute an estimate $f_{\alpha}$ of each $F(S_{\alpha})$ using $\Gamma^{k}$-circuits with $k=2^m+1$; more precisely, we compute $K$ numbers $f_{\alpha}$ satisfying $|f_{\alpha}- F(S_{\alpha})|\leq \epsilon/2$. The runtime of the computation will be poly$(n, 1/\epsilon)$ and the success probability exponentially close to 1. Finally,  we compute $c:= [\sum f_{\alpha}]/K$ which takes poly$(n, 1/\epsilon)$ time as well. Using (\ref{chernoff_bounded_depth}) and the triangle inequality  it follows that $|c - \gamma|\leq \epsilon$. Thus $c$ is our desired polynomial approximation of $\gamma$. \finpr

Finally we note that theorem \ref{thm_bounded_depth} can be generalized in the following rather intriguing sense: using $\Gamma^k$-circuits one can also efficiently estimate matrix elements of the form $|\langle 0|{\cal U}{\cal C}|0\rangle|^2$ where ${\cal U}$ is again a constant-depth circuit and where ${\cal C}$ represents an arbitrary uniform family of \emph{Clifford circuits}. Interestingly, these Clifford circuits need not have constant depth. The proof, which is given in appendix \ref{app_sect_counded_depth_cliff},  uses an argument analogous to the proof of theorem \ref{thm_bounded_depth} combined with the alternate Hadamard test given in lemma \ref{thm_improved_hadamard}.

\section{Acknowledgements}
The proof of Theorem \ref{thm_simulation_two_local} was found in a discussion with V. Murg and M. Schwarz while M. Van den Nest visited the University of Vienna in November 2009.

\appendix

\section{A generalization of theorem \ref{thm_bounded_depth}}\label{app_sect_counded_depth_cliff}

\begin{theorem}
Let ${\cal U}$ be a (uniform family of)  $n$-qubit quantum circuit(s) of depth $m$. Let ${\cal C}$ be a (uniform family of)  $n$-qubit Clifford circuit(s). Then the problem of estimating the matrix element $|\langle 0|{\cal U}{\cal C}|0\rangle|^2$  with polynomial accuracy and with success probability exponentially (in $n$) close to 1 is in  $\Gamma^k$  with $k=2^m+1$.
\end{theorem}
\begin{proof} Similar to (\ref{ME_bounded_depth1}) one has \be\label{appendix1} |\langle 0|{\cal C}{\cal U}|0\rangle|^2 = \frac{1}{2^n} \sum \langle 0|{\cal U}^{\dagger}{\cal C}^{\dagger}Z(S){\cal C} {\cal U}|0\rangle.\ee Since ${\cal C}$ is Clifford, ${\cal C}^{\dagger}Z(S){\cal C} =: P$ is a Pauli operator which can moreover be determined efficiently; that we  suppress dependence of $P$ on $S$ to simplify notation. Following (\ref{Pauli}), we can write   \be
P=i^{t}\prod X_k^{a_k}Z_k^{b_k},\quad \mbox{ where } t \in \{0,1,2,3\},\ a_k, b_k \in \{0,1\}.
\ee
Now define $G_k := {\cal U}^{\dagger}X_k^{a_k}{\cal U}$ and $H_k := {\cal U}^{\dagger}Z_k^{a_k}{\cal U}$ as well as ${\cal C}_1:= \prod G_k$ and ${\cal C}_2:= \prod H_k$. Then \be \langle 0|{\cal U}^{\dagger}{\cal C}^{\dagger}Z(S){\cal C} {\cal U}|0\rangle = i^{t} \langle 0|{\cal C}_1{\cal C}_2|0\rangle.\ee Since the $Z_k$ mutually commute, the $G_k$ mutually commute as well. Furthermore each $G_j$ acts on at most $2^m$ qubits. Therefore ${\cal C}_1$ is a $2^m$-local commuting circuit. Similarly, ${\cal C}_2$ is a $2^m$-local commuting circuit as well. We can now apply theorem \ref{thm_comm_non_comm}, showing that $\langle 0|{\cal C}_1{\cal C}_2|0\rangle$ can be estimated with polynomial accuracy using $\Gamma^k$-circuits with $k=2^m+1$. Continuing the argument as in the proof of theorem \ref{thm_bounded_depth} completes the proof.
\end{proof}

\bibliography{commuting}{}

\begin{thebibliography}{10}

\bibitem{shor1999polynomial}
P.~W. Shor (1999), {\em {Polynomial-time algorithms for prime factorization and
  discrete logarithms on a quantum computer}}, SIAM review, vol. 41, no.~2,
  pp.~303--332.

\bibitem{Jozsa2002Role}
R.~Jozsa and N.~Linden (2003), {\em {On the role of entanglement in quantum
  computational speed-up}}, Proc. R. Soc. A, vol. 459, pp.~2011--2032,
  quant-ph/0201143.

\bibitem{Vidal2003Efficient}
G.~Vidal (2003), {\em {Efficient classical simulation of slightly entangled
  quantum computations}}, Phys. Rev. Lett., vol. 91, p.~147902,
  quant-ph/0301063.

\bibitem{Gottesman1997Stabilizer}
D.~Gottesman (1997), {\em {Stabilizer Codes and Quantum Error Correction}},
  quant-ph/9705052.

\bibitem{Valiant2002Quantum}
L.~G. Valiant (2002), {\em {Quantum Circuits That Can Be Simulated Classically
  in Polynomial Time}}, SIAM J. Comput., vol. 31, pp.~1229--1254.

\bibitem{Knill2001Fermionic}
E.~Knill (2001), {\em {Fermionic Linear Optics and Matchgates}},
  quant-ph/0108033.

\bibitem{Jozsa2008Matchgates}
R.~Jozsa and A.~Miyake (2008), {\em {Matchgates and classical simulation of
  quantum circuits}}, Proc. R. Soc. A, vol. 464, pp.~3089--3106,
  arXiv:0804.4050.

\bibitem{VandenNest2010Simulating}
M.~Van~den Nest (2010), {\em {Simulating quantum computers with probabilistic
  methods}}, Quantum Inf. and Comp., vol. 11, pp.~784--812, arXiv:0911.1624.

\bibitem{VandenNest2012Efficient}
M.~Van~den Nest (2012), {\em {Efficient classical simulations of quantum
  Fourier transforms and normalizer circuits over Abelian groups}},
  arXiv:1201.4867.

\bibitem{Aaronson2011computational}
S.~Aaronson and A.~Arkhipov (2011), {\em The computational complexity of linear
  optics} in {\em Proceedings of the 43rd annual ACM symposium on Theory of
  computing}, pp.~333--342, ACM, arXiv:1011.3245.

\bibitem{Bremner2010Classical}
M.~J. Bremner, R.~Jozsa, and D.~J. Shepherd (2011), {\em {Classical simulation
  of commuting quantum computations implies collapse of the polynomial
  hierarchy}}, Proc. R. Soc. A, vol. 467, pp.~459--472, arXiv:1005.1407.

\bibitem{Jordan2010permutational}
S.~Jordan (2010), {\em Permutational quantum computing}, Quantum Inf. and
  Comp., vol. 10, no.~5, pp.~470--497, arXiv:0906.2508.

\bibitem{Briegel2001persistent}
H.~Briegel and R.~Raussendorf (2001), {\em Persistent entanglement in arrays of
  interacting particles}, Physical Review Letters, vol. 86, no.~5,
  pp.~910--913, quant-ph/0004051.

\bibitem{Bravyi2004Commutative}
S.~Bravyi and M.~Vyalyi (2005), {\em {Commutative version of the k-local
  Hamiltonian problem and common eigenspace problem}}, Quantum Inf. and Comp.,
  vol. 5, pp.~187--215, quant-ph/0308021.

\bibitem{Aharonov2011Complexity}
D.~Aharonov and L.~Eldar (2011), {\em On the complexity of Commuting Local
  Hamiltonians, and tight conditions for Topological Order in such systems} in
  {\em Foundations of Computer Science (FOCS), 2011 IEEE 52nd Annual Symposium
  on}, pp.~334--343, IEEE, arXiv:1102.0770.

\bibitem{Schuch2011complexity}
N.~Schuch (2011), {\em Complexity of commuting Hamiltonians on a square lattice
  of qubits}, Quantum Information and Computation, vol. 11, no.~11-12,
  pp.~901--912, arXiv:1105.2843.

\bibitem{Hastings2012trivial}
M.~Hastings (2012), {\em Trivial Low Energy States for Commuting Hamiltonians,
  and the Quantum PCP Conjecture}, arXiv:1201.3387.

\bibitem{Shepherd2009instantaneous}
D.~Shepherd and M.~J. Bremner (2009), {\em Temporally unstructured quantum
  computation}, Proc. R. Soc. A, vol. 465, pp.~1413--1439, arXiv:0809.0847.

\bibitem{Shepherd2010binary}
D.~Shepherd (2010), {\em Binary Matroids and Quantum Probability
  Distributions}, arXiv:1005.1744.

\bibitem{Hayden2004randomizing}
P.~Hayden, D.~Leung, P.~Shor, and A.~Winter (2004), {\em Randomizing quantum
  states: Constructions and applications}, Commun. Math. Phys., vol. 250,
  no.~2, pp.~371--391, quant-ph/0307104.

\bibitem{Nielsen2000quantum}
M.~A. Nielsen and I.~L. Chuang (2000), {\em Quantum computation and quantum
  information}.
\newblock Cambridge University Press.

\bibitem{Dehaene2003the}
J.~Dehaene and B.~De~Moor (2003), {\em The Clifford group, stabilizer states,
  and linear and quadratic operations over {GF}(2)}, Phys. Rev. A, vol. 68,
  p.~042318, quant-ph/0304125.

\end{thebibliography}
\bibliographystyle{hquantinf}

\end{document}